\documentclass[12pt,a4paper]{amsart}

\usepackage[latin1]{inputenc}     
\usepackage{amsmath,amssymb,amsthm}
\usepackage{latexsym}
\usepackage{amsfonts}
\usepackage{bbm,bm,dsfont} 
\usepackage{color} 
\usepackage{graphicx}

\numberwithin{equation}{section} 

\theoremstyle{definition}

\newtheorem{definition}{Definition}
\newtheorem{example}{Example}
\newtheorem{remark}{Remark}
\newtheorem{theorem}{Theorem}
\newtheorem{lemma}{Lemma}
\newtheorem{corollary}{Corollary}

\newcommand{\sfa}{\mathsf{A}}

\newcommand{\sfe}{\mathsf{E}}
\newcommand{\sff}{\mathsf{F}}

\newcommand{\sfo}{\mathsf{O}}

\newcommand{\N}{\mathbb N} 
\newcommand{\R}{\mathbb R} 
\newcommand{\C}{\mathbb C} 

\newcommand{\vp}{\varphi} 
\newcommand{\Om}{\Omega} 

\newcommand{\hi}{\mathcal{H}} 
\newcommand{\hil}{{\mathcal H}} 
\newcommand{\ki}{\mathcal{K}} 
\newcommand{\li}{\mathcal{L}} 

\renewcommand{\O}{\mathrm{Obs}} 

\newcommand{\lh}{\mathcal{L(H)}} 
\newcommand{\lk}{\mathcal{L(K)}} 
\newcommand{\sh}{\mathcal{S(H)}} 

\newcommand{\tr}[1]{\mathrm{tr}\left[#1\right]} 
\def\<{\langle} 
\def\>{\rangle} 

\newcommand{\Ao}{\mathsf{A}} 
\newcommand{\Po}{\mathsf{P}} 
\newcommand{\Mo}{\mathsf{M}} 
\newcommand{\No}{\mathsf{N}} 

\def\d{{\mathrm d}} 
\newcommand{\ov}{\overline} 

\newcommand{\mr}[1]{\mathrm{#1}}
\newcommand{\mc}[1]{\mathcal{#1}}

\newcommand{\ms}[1]{\mathsf{#1}}

\newcommand{\f}{\varphi}

\begin{document}

\title{Compatibility properties of extreme quantum observables}

\author{Erkka Haapasalo}
\email{ethaap@utu.fi}
\address{Turku Centre for Quantum Physics, Department of Physics and Astronomy, University of Turku, FI-20014 Turku, Finland}

\author{Juha-Pekka Pellonp\"a\"a}
\email{juhpello@utu.fi}
\address{Turku Centre for Quantum Physics, Department of Physics and Astronomy, University of Turku, FI-20014 Turku, Finland}

\author{Roope Uola}
\email{roope.uola@gmail.com}
\address{Naturwissenschaftlich-Technische Fakult\"at,
Universit\"at Siegen,
Walter-Flex-Str.~3,
57068 Siegen, Germany}
\address{Turku Centre for Quantum Physics, Department of Physics and Astronomy, University of Turku, FI-20014 Turku, Finland}

\begin{abstract}
Recently a problem concerning the equivalence of joint measurability and coexistence of quantum observables was solved \cite{lyhyt}. In this paper we generalize two known joint measurability results from sharp observables to the class of extreme observables and study relationships between coexistence, joint measurability, and post-processing of quantum observables when an extreme observable is involved. We also discuss another notion of compatibility and provide a counterexample separating this from the former notions.
\end{abstract}
\maketitle


\section{Introduction}

The problem of simultaneous measurement of observables lies in the heart of quantum mechanics. Two basic notions of simultaneous measurability are coexistence and joint
measurability. These notions have been studied widely \cite{Lahti, LaPu, Ludwig, Pell2} and recently Reeb {\it et al}.\ demonstrated that the notion of coexistence is more general than joint measurability \cite{lyhyt}.

The intuitive idea behind these definitions is quite simple. In the joint measurability case one is asked if there exists an observable giving two fixed observables ($\ms A$ and $\ms B$) as its marginals. In the coexistence case the marginal requirement is dropped -- it suffices that there exist an observable whose range contains the ranges of $\ms A$ and $\ms B$. An equivalent way to formulate coexistence is to require the existence of an observable which gives each pair of binarizations of $\ms A$ and $\ms B$ as post-processings. This formulation connects coexistence to joint measurability and thus one can see joint measurability as the main compatibility property of observables. A bit less restrictive (and maybe a bit more physical) question would be whether for each pair of binarizations there exists a joint observable. We call this property {\it joint measurability of binarizations}. Physically this means that in the coexistence case all the binarizations are jointly measurable with one measurement device and varying post-processings whereas in the other case one might need a different measurement device for each pair of binarizations.

In the coexistence case one has an observable $\ms M$ whose measurement statistics include the statistics of $\ms A$ and $\ms B$, although in this case one may not have a way to connect the statistics of the observable $\ms M$ to the statistics of $\ms A$ or $\ms B$. In this paper we show that if one of the concerned observables ($\ms A$, $\ms B$ or the mother observable $\ms M$) is extreme the connection between the statistics can be resolved. Extremality of an observable has already been connected to compatibility properties of quantum observables \cite{HaHePe13} and this work is a continuation of this line of research.

It is clear that the following implications hold:
\begin{align*}
\text{Joint measurability }\Longrightarrow\text{Coexistence }\Longrightarrow\text{Joint measurability of binarizations.}
\end{align*}
It was shown in \cite{lyhyt} that the first implication cannot be reversed and in \cite{2_coex} that joint measurability of binarizations does not imply joint measurability. In this paper we complete the picture by showing that the last implication is also irreversible.

The mathematical framework and basic definitions are included in the second section. The third, fourth and fifth sections deal with our new results and conclusions are found in the sixth section.

\section{Notations and definitions}

For any Hilbert space $\hi$ we let $\lh$ denote the set of bounded operators on $\hi$. The set of states (positive operators of trace one) is denoted by $\sh$ and $I_\hi$ (or briefly $I$) is the identity operator of $\hi$. A positive operator $E\in\lh$ is said to be an effect if $E\le I_\hi$. Throughout this article, we let $\hi$ be a separable complex nontrivial Hilbert space. Moreover, we say that $(\Omega,\Sigma)$ is a measurable space when $\Om$ is a non-empty set and $\Sigma$ is a $\sigma$-algebra of subsets of $\Omega$. A measurable space $(\Om,\Sigma)$ is said to be standard Borel, if it is $\sigma$-isomorphic to a Borel-measurable subset $B$ of a Polish space $E$ equipped with the restriction $\mc B(B)=\{B\cap C\,|\,C\in\mc B(E)\}$ of the Borel $\sigma$-algebra $\mc B(E)$ of $E$. The total set $\Om$ of a standard Borel space can only be finite, countably infinite, or of the cardinality of the continuum $\R$. In the first two cases, the standard Borel space is isomorphic to some $Z\subseteq\N:=\{1,2,3,\ldots\}$ (equipped with the discrete $\sigma$-algebra) and, in the last case, the standard Borel space is isomorphic to $\big(\R,\mc B(\R)\big)$. All physically relevant value spaces fall in the category of standard Borel
spaces. For instance, any finite-dimensional second-countable manifold which is Hausdorff is standard Borel when equipped with its Borel structure.

For a measurable space $(\Om,\Sigma)$ let $\O(\Sigma,\,\hi)$ be the convex set of quantum observables, that is, normalized positive operator valued measures (POVMs)
$\ms A:\,\Sigma\to\lh$. Recall that a map $\ms A:\,\Sigma\to\lh$ is a POVM if and only if $X\mapsto p_\rho^\Ao(X):=\tr{\rho\ms A(X)}$ is a probability measure for all $\rho\in\sh$. Moreover, $p_\rho^\Ao(X)$ is interpreted as the probability of getting an outcome $x$ which belong to $X\in\Sigma$ when a measurement of $\ms A\in\O(\Sigma,\,\hi)$ is performed and the system is in the state $\rho\in\sh$. Any $\ms A\in\O(\Sigma,\,\hi)$ is called a projection valued measure (PVM) or a sharp observable if $\ms A(X)^2=\ms A(X)$ for all $X\in\Sigma$.

In the rest of this section, we let $(\Omega,\Sigma)$ and $(\Omega',\Sigma')$ be measurable spaces and $\ms A\in\O(\Sigma,\hi)$ and $\ms B\in\O(\Sigma',\hi)$.

\begin{definition}[Binarization]
Any effect $E\in\lh$ defines a two-valued (binary) observable $\sfo^E:\,\big\{\emptyset,\,\{+1\},\,\{-1\},\,\{+1,-1\}\big\}\to\lh$ via $\sfo^E(\{+1\}):=E$ (so that automatically $\sfo^E(\{-1\})=I_\hi-E$). For each $X\in\Sigma$ we can define a {\it binarization} of  $\ms A$ (associated to $X$) as a binary observable  $\sfo^{\ms A(X)}$. 
\end{definition}

\begin{definition}[Commutativity of POVMs]
If $[\ms A(X),\ms B(Y)]=0$ for all $X\in\Sigma$, $Y\in\Sigma'$, then $\ms A$ and $\ms B$ are said to {\it commute} (with each other).
\end{definition}

\begin{definition}[Joint measurability]
Denote the product $\sigma$-algebra of $\Sigma$ and $\Sigma'$ by $\Sigma\otimes\Sigma'$. We say that $\ms A$ and $\ms B$ are {\it jointly measurable} if there exists an $\No\in\O(\Sigma\otimes\Sigma',\hi)$ such that $\No(X\times\Omega')=\ms A(X)$ for all $X\in\Sigma$ and $\No(\Omega\times Y)=\ms B(Y)$ for all $Y\in\Sigma'$. In this case, $\No$ is called a {\it joint observable} of $\ms A$ and $\ms B$. Recall that if $\ms A$ and $\ms B$ commute then $\ms A$ and $\ms B$ are jointly measurable.\footnote{Actually, the value spaces $(\Om,\Sigma)$ and $(\Om',\Sigma')$ must satisfy certain mathematical requirements, e.g.\ they are standard Borel spaces \cite{ylinen96}.} 
\end{definition}

\begin{definition}[Coexistence]\label{17.3.2014 14:37}
We say that $\ms A$ and $\ms B$ are {\it coexistent} if there exists a $\sigma$-algebra $\ov\Sigma$ over a set $\ov\Om$ and an observable $\Mo:\,\ov\Sigma\to\lh$ such that the ranges\footnote{Recall that the range of $\ms A$ is the set ${\rm ran}\,\ms A:=\{\ms A(X)\in\lh\,|\,X\in\Sigma\}$.} of $\ms A$ and $\ms B$ are contained in the range of $\Mo$ (i.e.,\ ${\rm ran}\,\ms A\cup{\rm ran}\,\ms B\subseteq{\rm ran}\,\Mo$). In other words, for any $X\in\Sigma$ and $Y\in\Sigma'$ there exist sets $Z,\,Z'\in\ov\Sigma$ such that $\Mo(Z)=\ms A(X)$ and $\Mo(Z')=\ms B(Y)$. In this case, we say that $\Mo$ is a {\it mother observable} of $\ms A$ and $\ms B$.
\end{definition}

\begin{definition}[Smearing]\label{def:smearing}
Let $(\ov\Om,\ov\Sigma)$ be a measurable space and $\mu$ be a $\sigma$-finite positive measure on $\ov\Sigma$. We say that $\beta:\,\ov\Omega\times\Sigma\to\R$ is a {\it weak Markov kernel (with respect to $\mu$)} if
\begin{itemize}
\item[(i)] $\ov\Omega\ni z\mapsto\beta(z,X)\in\R$ is $\mu$-measurable for all $X\in\Sigma$,
\item[(ii)] for every $X\in\Sigma$, $0\le\beta(z,X)\le 1$ for $\mu$-almost all $z\in\ov\Omega$,
\item[(iii)] $\beta(z,\Omega)=1$ and $\beta(z,\emptyset)=0$ for $\mu$-almost all $z\in\ov\Omega$,
\item[(iv)] if $\{X_i\}_{i=1}^\infty\subseteq\Sigma$ is a disjoint sequence (i.e.\ $X_i\cap X_j=\emptyset$, $i\ne j$) then $\beta(z,\cup_i X_i)=\sum_i\beta(z,X_i)$ for $\mu$-almost all $z\in\ov\Omega$.
\end{itemize}
Moreover, $\Mo\in\O(\ov\Sigma,\hil)$ is {\it absolutely continuous with respect to $\mu$} if $\mu(Z)=0$ implies $\Mo(Z)=0$ for all $Z\in\ov\Sigma$. If there exists a $\sigma$-finite measure $\mu:\,\ov\Sigma\to[0,\infty]$ such that $\Mo$ is absolutely continuous with respect to it and a weak Markov kernel $\beta:\,\ov\Omega\times\Sigma\to\R$ with respect to $\mu$ such that $\ms A(X)=\int_{\ov\Omega}\beta(z,X)\d\Mo(z)$ for all $X\in\Sigma$ then $\ms A$ is a {\it smearing} or a {\it post-processing} of $\Mo$. A special case of a weak Markov kernel is a {\it Markov kernel} which is a map $\beta:\ov\Omega\times\Sigma\to\R$ such that $\beta(z,\cdot):\Sigma\to\R$ is a probability measure for ($\mu$-almost) all $z\in\ov\Omega$. Note that if $(\Omega,\Sigma)$ is a standard Borel space and $\Ao$ is a smearing of $\Mo$ then we may always replace the weak Markov kernel by a Markov kernel \cite[Theorem 6.3 and the subsequent discussion]{JePuVi}.

Let $\ms A\in\O(\Sigma,\hil)$ (resp. $\ms B\in\O(\Sigma',\hil)$) be a smearing of $\ms M\in\O(\ov\Sigma,\hil)$ by means of a (weak) Markov kernel $\beta:\ov\Om\times\Sigma\to\R$ (resp. $\gamma:\ov\Om\times\Sigma'\to\R$), where $(\Om,\Sigma)$ and $(\Om',\Sigma')$ are standard Borel. One may easily check that the function
$$
\Sigma\times\Sigma'\ni(X,Y)\mapsto\int_{\ov\Om}\beta(z,X)\gamma(z,Y)\d\ms M(z)
$$
is a positive-operator bimeasure which we may extend into a joint observable $\ms N\in\O(\Sigma\otimes\Sigma',\hil)$ for $\ms A$ and $\ms B$ \cite{LaYl, ylinen96}. Hence, $\ms A$ and $\ms B$ are jointly measurable.

When $\ms A\in\O(\Sigma,\hil)$ is obtained as a post-processing of $\ms M\in\O(\ov\Sigma,\hil)$ by means of a Markov kernel $\beta:\ov\Omega\times\Sigma\to\R$ of the form $\beta(z,X)=\chi_{f^{-1}(X)}(z)$, $z\in\ov\Om$, $X\in\Sigma$, where $f:\ov\Om\to\Om$ is a $(\ov\Sigma,\Sigma)$-measurable function, we say that $\ms A$ is a {\it relabeling} of $\ms M$; here $\chi_S$ stands for the characteristic or indicator function of a set $S$. Now $\ms A(X)=\int_{\ov\Om}\beta(z,X)\d\Mo(z)=\Mo(f^{-1}(X))$, $X\in\Sigma$.
\end{definition}

\begin{definition}[Discrete observable]

We say that $\ms A$ is {\it discrete} if there exists a finite or countably infinite set $\Xi\subseteq\Omega$ so that $\ms A$ is absolutely continuous with respect to the measure $\sum_{x\in\Xi}\delta_x$ (where $\delta_x$ is the Dirac measure concentrated on the point $x$). This implies that one can identify $\ms A$ with the sequence $(\ms A_i)_{i=1}^N$ of effects, where $\ms A_i:=\ms A(X_i)\ne 0$, $N\in\N\cup\{\infty\}$, and  $\{X_i\}_{i=1}^N\subseteq\Sigma$ is some disjoint collection of sets. In practice, we may restrict $\ms A$ to the sub-$\sigma$-algebra generated by $\{X_i\}_{i=1}^N$.
\end{definition}

\begin{remark}\label{remarkinpoikanen}
It is easy to see that the joint measurability of $\ms A$ and $\ms B$ implies their coexistence but the converse does not hold in general \cite[Proposition 1]{lyhyt}. However, two coexistent binary observables are jointly measurable \cite{Lahti}. Indeed this is true also for $n\in\N$ binary observables $(\sfo^{E_i})_{i=1}^n$ (where $E_i$'s are effects). To prove this assume that the observables $\sfo^{E_i}$ are coexistent in the sense that there exist (a $\sigma$-algebra $\ov\Sigma$ over a set $\ov\Omega$ and) an observable $\Mo\in\O(\ov\Sigma,\hi)$ and sets $Z_i\in\ov\Sigma$ such that $E_i=\Mo(Z_i)$ for all $i=1,\ldots, n.$ Then one can define a joint observable $\ms N$ (defined on the power set of the $n$-fold cartesian product $\{+1,-1\}\times\cdots\times\{+1,-1\}$) for the observables $(\sfo^{E_i})_{i=1}^n$ by
\begin{align*}
\ms N\big(\{(x_1,...,x_n)\}\big)=\Mo(Z_1^{x_1}\cap Z_2^{x_2}\cap...\cap Z_n^{x_n}),
\end{align*}
where $x_i\in\{-1,1\}$ and $Z_i^+:=Z_i$, $Z_i^-:=\ov\Om \setminus Z_i$. Note that, in general, defining a joint observable from the mother observable as above is not possible. But for two valued observables it is easy to check that the map $\ms N$ above is a joint observable for observables $\sfo^{E_i}$, that is,
$$
\sfo^{E_i}(\{\pm1\})=\Mo(Z_i^\pm)=\sum_{j\ne i\atop x_j=\pm1}\ms N\big(\{(x_1,...,x_{i-1},\pm1,x_{i+1},...,x_n)\}\big).
$$
We will see in sections \ref{sec_ext} and \ref{sec_mot} that a similar construction works for extreme observables.

If $\ms A$ and $\ms B$ are coexistent (with a mother $\Mo$) then their binarizations $\sfo^{\ms A(X)}$ and $\sfo^{\ms B(Y)}$ are (coexistent and) jointly measurable for all $X\in\Sigma$ and $Y\in\Sigma'$. In this case, one can choose the same mother observable $\Mo$ for all pairs $\big(\sfo^{\ms A(X)},\,\sfo^{\ms B(Y)}\big)$. Moreover, $\sfo^{\ms A(X)}\big(\{+1\}\big)=\ms A(X)=\Mo(Z)=\int_{\ov\Om}\beta(z,\{+1\})\d\Mo(z)$ for some $Z\in\ov\Sigma$, where $\beta(z,\{+1\})=\chi_Z(z)$. Clearly,
$\beta(z,\{+1\})$ extends to a Markov kernel $\beta(z,\{\pm1\})=\chi_{f^{-1}(\{ \pm1\})}(z)$, where $f:\ov\Om\to\{+1,-1\}$ is a measurable function such that
$f^{-1}(\{+1\})=Z$ and $f^{-1}(\{-1\})=\ov\Om\setminus Z$, and thus $\sfo^{\ms A(X)}$ is a smearing of $\ms M$. Similarly, any $\sfo^{\ms B(Y)}$ is a smearing of $\Mo$.
\end{remark}

\section{Joint measurability of binarizations $\not\Rightarrow$ coexistence}

From Remark \ref{remarkinpoikanen} it is clear that if two observables are coexistent then there exists an observable from which one can post-process all the binarizations of these observables. This means that if the observables $\ms A$ and $\ms B$ are coexistent then there exists an observable $\Mo$ such that for every sets $X$ and $Y$ the binarizations $\sfo^{\ms A(X)}$ and $\sfo^{\ms B(Y)}$ are jointly measurable in the sense that they can be post-processed from the observable $\Mo$. More generally, one could require that for every sets $X$ and $Y$ there exists an observable $\Mo_{X,Y}$ such that the binarizations $\sfo^{\ms A(X)}$ and $\sfo^{\ms B(Y)}$ can be post-processed from $\Mo_{X,Y}$. The following theorem shows that the latter property, which we call {\it joint measurability of binarizations}, is more general than coexistence.

\begin{theorem}\label{25.2.2015 13:26}
Joint measurability of binarizations does not imply coexistence.
\end{theorem}
\begin{proof}
Consider the case $\mathcal H =\mathbb C ^2$. Let  $\varphi_1=(1,0)$, $\varphi_2=(0,1)$, and $\psi=\frac{1}{\sqrt 2}(\varphi_1-\varphi_2)$. The effects $E_1=\frac{4}{7}|\varphi_1\rangle\langle\varphi_1|$, $E_2=\frac{4}{7}|\varphi_2\rangle\langle\varphi_2|$ and $E_3=I-E_1-E_2$ define a 3-valued observable $\sfe$ and the effects $F_1=\frac{4}{7}|\psi\rangle\langle\psi|$ and $F_2=I-F_1$ constitute a 2-valued observable $\sff$.

If the observables $\sfe$ and $\sff$ are coexistent there exist a (mother) observable $\Mo:\,\ov\Sigma\to\mathcal L(\mathbb C^2)$ such that ${\rm ran}\,\sfe\cup{\rm ran}\,\sff\subseteq{\rm ran}\,\Mo$. This implies that there exist sets $X,\,Y,\,Z\in\ov\Sigma$ such that $E_1=\Mo(X),$ $E_2=\Mo(Y)$ and $F_1=\Mo(Z)$. Since the effects $E_1,\, E_2$ and $F_1$ are rank-1 and for example $\Mo(X\cap Y)\leq E_1,\,E_2$ one must have $\Mo(X\cap Y)=0$. Because $\Mo$ gives zero for all the pairwise intersections of the sets $X,$ $Y,$ and $Z$, and $\Mo(X\cap Y\cap Z)\leq\Mo(X\cap Y)=0$ we get
\begin{align}\label{2.9.2013 15:55}
\Mo(X\cup Y\cup Z)=E_1+E_2+F_1\leq I.
\end{align}
This does not hold since the greatest eigenvalue of $E_1+E_2+F_1$ is $\frac{8}{7}$. Hence, the observables $\sfe$ and $\sff$ are not coexistent.

The binarizations of observables $\sfe$ and $\sff$ are jointly measurable if all the binarizations of $\sfe$ are coexistent with $\sff$ (since for binary observables coexistence is equivalent to joint measurability). Sufficient for this are the following three conditions:
\begin{align}\label{2.9.2013 16:25}
E_1+F_1&\leq I\nonumber\\
E_2+F_1&\leq I\\
E_3+F_1&\leq I\nonumber.
\end{align}
For example, we have a binarization $\sfo^{E_1}$ whose range is $\{E_1, I-E_1\}$. If $E_1+F_1\leq I$ holds one can define a mother observable $\Mo:\,\ov\Sigma\to\mathcal L(\mathbb C^2)$ by $\Mo(X)=E_1$, $\Mo(Y)=F_1$ for some disjoint $X,Y\in\ov\Sigma$.

By calculating the eigenvalues of the above operators $E_i+F_1$ one can see that the conditions $E_i+F_1\leq I$ are valid.
\end{proof}

\begin{remark}
Recently \cite{RU, MQ} an equivalence between joint measurability and the impossibility of quantum steering was shown. Quantum steering is a bipartite entanglement verification method where one party (Alice) tries to convince the other party (Bob) that their shared state is entangled by making only local measurements on her system. For this remark it is enough to formulate steering for discrete observables: Let Alice and Bob share a quantum state $\rho_{AB}\in\mathcal S(\hil\otimes\hil)$ and let the measurements performed by Alice be labelled by $\{\ms A_k\}_{k=1}^n$. When Alice makes a (L\"uders) measurement of $\ms A_k$ and records an outcome $x$, the post-measurement state of the bipartite system will be
\begin{align}
\sigma_{x|k}:=\text{tr}_A[(\ms A_k(x)\otimes I)\rho_{AB}].
\end{align}
We say that Alice can steer Bob, if Bob can not reproduce the conditional states $\sigma_{x|k}$ from some local ensemble of positive operators $\{\rho_\lambda\}_\lambda\subset\mathcal S(\hil)$ by the means of classical post-processing.

It is easy to see that for separable states, i.e. states of the form $\rho_{AB}=\sum_i \mu_i\rho_i^A\otimes\rho_i^B$ with non-negative coefficients $\mu_i$, $\sum_i\mu_i=1$, and $\rho_i^A,\rho_i^B\in\mathcal S(\hil)\ \forall i$, steering is not possible as one can choose the local ensemble to be $\{\mu_i\rho_i^B\}_i$ and for given parameters $i,k,x$ the Markov kernels (stochastic matrices in the discrete case) to be $p(x|k,i):=\text{tr}[\ms A_k(x)\rho_i^A]$.

Formally this means that Alice can steer Bob if and only if there is no ensemble of positive operators $\{\rho_\lambda\}$ together with a set of suitable stochastic matrices $p(x|k,\lambda)$ such that
\begin{align}
\sigma_{x|k}=\sum_\lambda p(x|k,\lambda)\rho_\lambda.
\end{align}
The results of \cite{RU, MQ} state that if Alice uses jointly measurable observables she can never steer Bob no matter what the shared state is, and, moreover, that with non-jointly measurable observables there always exists a state with which steering is present.

Theorem \ref{25.2.2015 13:26} points out an interesting detail about quantum steering: consider any of the binarizations of the observable $\sfe$ used in the proof of theorem \ref{25.2.2015 13:26}. If Alice uses one of these binarizations together with the observable $\sff$ given in the same proof, steering will never be possible, i.e. the scenario is classical in this sense. However, if Alice measures the observable $\sfe$ instead of its binarization, she will be able to reach the quantum regime, i.e. she will be able to steer Bob. In other words, Theorem \ref{25.2.2015 13:26} shows that there are occasions where the quantum phenomenon of steering will become accessible when Alice makes more detailed measurements on her system.

\end{remark}

\section{Coexistence with an extreme observable}\label{sec_ext}

We denote a minimal Na\u{\i}mark dilation of an observable $\ms A\in\O(\Sigma,\hi)$ by a triple $(\mc K,\ms P,J)$, where $\mc K$ is a dilation (Hilbert) space, $J:\hi\to\mc K$ is an isometry and $\Po\in\O(\Sigma,\mc K)$ is a sharp observable. The dilation is defined by the formula
$$
\ms A(X)=J^*\Po(X) J,\qquad X\in\Sigma,
$$
and the minimality means that the linear hull of the vectors of the form $\ms P(X)J\f$, $X\in\Sigma$, $\f\in\hil$, is a dense subspace of $\mc K$. Note that if $(\Om,\Sigma)$ is standard Borel then $\ki$ is separable.

\begin{remark}
Let $(\mc K,\ms P,J)$ be a minimal Na\u{\i}mark dilation for $\ms A\in\O(\Sigma,\hi)$. $\Mo$ is an {\it extreme}\footnote{An observable $\ms A\in\O(\Sigma,\hil)$ is an extreme point of the convex set $\O(\Sigma,\hil)$ if for any $t\in(0,1)$ and $\ms A_1,\,\ms A_2\in\O(\Sigma,\hil)$ the condition $\ms A=t\ms A_1+(1-t)\ms A_2$ yields $\ms A_1=\ms A_2=\ms A$. The convex combination of observables is defined by $\big(t\ms A_1+(1-t)\ms A_2\big)(X):=t\ms A_1(X)+(1-t)\ms A_2(X)$ for all $X\in\Sigma$.} point of $\O(\Sigma,\hi)$ if and only if, for all $D\in\mathcal{L}(\mc K)$ such that $[D,\ms P(X)]=0$, $X\in\Sigma$, the condition $J^*DJ=0$ implies $D=0$ \cite{Arveson69, VIITATKAA, Part1}. Recall that sharp observables are extreme, but there are other extreme observables as well \cite{VIITATKAA}.
\end{remark}

\begin{lemma}\label{lemma1}
Let $A:\,\hi\to\ki$ and $B\in\lh$ be bounded operators. Then $0\le B\le A^*A$ if and only if there exists a $C\in\lk$, $0\le C\le I_{\ki}$, such that $B=A^*CA$. In addition, $C$ is unique if and only if the range ${\rm ran}\,A$ of $A$ is dense in the Hilbert space $\ki$.
\end{lemma}

\begin{proof}
For any $\psi\in\hi$ define
$
C(A\psi):=\sqrt{B}\psi.
$
If $A\psi=A\psi'$, i.e.\ $A\psi_{-}=0$, $\psi_-:=\psi-\psi'$, then
$$
0\le \|C(A\psi_-)\|^2=\|\sqrt{B}\psi_-\|^2=\<\psi_-|B\psi_-\>\le\<\psi_-|A^*A\psi_-\>=\|A\psi_-\|^2=0
$$
so that $C(A\psi)=C(A\psi')$ and $C$ is well defined. If $\eta\in({\rm ran}\,A)^\perp$ then define $C\eta:=0$. The rest is trivial.
\end{proof}

For the rest of this section, $(\Om,\Sigma)$, $(\Om',\Sigma')$, and $(\ov\Om,\ov\Sigma)$ are measurable spaces.

\begin{theorem}\label{lause_ykkonen}
Any discrete and extreme $\ms A\in\O(\Sigma,\hi)$ and any $\Mo\in\O(\ov\Sigma,\hi)$ such that ${\rm ran}\,\ms A\subseteq{\rm ran}\,\Mo$ are jointly measurable.
\end{theorem}

\begin{proof}
Let $\ms A\in\O(\Sigma,\hi)$ and $\Mo\in\O(\ov\Sigma,\hi)$ and assume that there are $X\in\Sigma$ and $Z_X\in\ov\Sigma$ such that $\Mo(Z_X)=\ms A(X)$. Let us pick a minimal Na\u{\i}mark dilation $(\mc K,\ms P,J)$ for $\ms A$. For all $Z\in\ov\Sigma$,
$$
\Mo(Z\cap Z_X)\le\Mo(Z_X)=\ms A(X)=[\ms P(X)J]^*\ms P(X)J
$$
and lemma \ref{lemma1} implies that, there exists a $C_{X}(Z)\in\li(\mc K)$ such that $0\le C_{X}(Z)\le \ms P(X)$ and 
\begin{equation}\label{kaava_kakkonen}
\Mo(Z\cap Z_X)=[\ms P(X)J]^*C_{X}(Z)\ms P(X)J=J^*\ms P(X)C_{X}(Z)\ms P(X)J=J^*C_{X}(Z)\ms P(X)J.
\end{equation}

Assume now that $\ms A$ is discrete, that is, there exists a disjoint\footnote{The sequence $\{X_i\}_{i=1}^\infty\subseteq\Sigma$ can contain empty sets and it may happen that $\sfa_i=0$ for some $i$'s.} sequence $\{X_i\}_{i=1}^\infty\subseteq\Sigma$ such that one can identify $\ms A$ with the sequence $\{\ms A_i\}_{i=1}^\infty$ where $\ms A_i:=\ms A(X_i)$. Similarly, we denote $\ms P_i=\ms P(X_i)$. Moreover, assume that $\ms A$ is extreme and ${\rm ran}\,\ms A\subseteq{\rm ran}\,\Mo$. Now, for all $X\in\Sigma$ there is $Z_X\in\ov\Sigma$ such that $\ms M(Z_X)=\ms A(X)$ and, as above, we have the positive operators $C_X(Z)\leq\ms P(X)$ such that (\ref{kaava_kakkonen}) holds for all $X\in\Sigma$ and $Z\in\ov\Sigma$. From now on, denote $C_i(Z):=C_{X_i}(Z)$ and $Z_i=Z_{X_i}$ for any $i$ and $Z\in\ov\Sigma$; it follows that $0\leq C_i(Z)\leq\ms P_i$ for any $i$ and $Z\in\ov\Sigma$. For any $X_i\ne X_j$ such that $\ms A_i\ne0\ne\ms A_j$ one has from equation \eqref{kaava_kakkonen} 
$$
\Mo(Z_i\cap Z_j)=J^*C_i(Z_j)\Po_iJ=J^*C_j(Z_i)\Po_jJ.
$$
By defining $D:=C_i(Z_j)\Po_i-C_j(Z_i)\Po_j$ one gets
$$
J^*DJ=0,\qquad [D,\ms P(X)]=0\quad\forall X\in\Sigma,
$$
implying $D=0$ so that, since the operators $C_i(Z_j)\Po_i=C_i(Z_j)$ and $C_j(Z_i)\Po_j=C_j(Z_i)$ are supported on orthogonal subspaces, we have $C_i(Z_j)=0=C_j(Z_i)$ and 
$$
\Mo(Z_i\cap Z_j)=0.
$$

From \eqref{kaava_kakkonen} we get
$$
\Mo(Z\cap Z_i)=J^*C_i(Z)\Po_iJ=J^*C_i(Z)J
$$
where $C_i:\,\ov\Sigma\to\li(\Po_i\mc K)$ is a unique POVM. Indeed, assume that $(W_n)_{n=1}^\infty$ is a disjoint sequence in $\ov\Sigma$. Denoting $W^N=\bigcup_{n=1}^NW_n$ for any $N=1,\,2,\ldots$, one has the equality
$$
J^*C_i(W^N)J=\Mo(W^N\cap Z_i)=\sum_{n=1}^N\Mo(W_n\cap Z_i)=\sum_{n=1}^NJ^*C_i(W_n)J
$$
for any $i$. Since the operator $C_i(W^N)-\sum_{n=1}^NC_i(W_n)$ commutes with the spectral measure $\ms P$, the previous equality together with the extremality of
$\ms A$ imply that $\sum_{n=1}^NC_i(W_n)=C_i(W^N)\leq\ms P_i$. The sequence $\big(\sum_{n=1}^NC_i(W_n)\big)_{N=1}^\infty=\big(C_i(W^N)\big)_{N=1}^\infty$ is thus an increasing (since all the summands are positive) sequence bounded from above by $\ms P_i$ implying that $\sum_{n=1}^\infty C_i(W_n):=w-\lim_{N\to\infty}\sum_{n=1}^NC_i(W_n)=\sup_{N\in\N}C_i(W^N)$ is defined, and since the limit is weak, the operator $\sum_{n=1}^\infty C_i(W_n)$ is in the commutant of the range of $\ms P$. As above, one now sees $C_i\big(\cup_{n=1}^\infty W_n\big)=\sum_{n=1}^\infty C_i(W_n)$ since the operator $C_i\big(\cup_{n=1}^\infty W_n\big)-\sum_{n=1}^\infty C_i(W_n)$ commutes with the $\ms P_j$'s. Hence, for all $i$, $C_i:\ov\Sigma\to\li(\Po_i\mc K)$ is weakly $\sigma$-additive. Similarly, $C_i(\ov\Om)=\Po_i$. Note that $\mc K=\bigoplus_{i=1}^\infty\Po_i\mc K$ and $C_i(\ov\Omega)=\Po_i=I_{\Po_i\mc K}$. Now one can define a joint observable $\No:\,\Sigma\otimes\ov\Sigma\to\lh$ via
\begin{equation}\label{5.2.2014 12:32}
\No(X_i\times Z):=J^*C_i(Z)J,
\end{equation}
which proves the claim.
\end{proof}

\begin{corollary}\label{5.2.2014 13:05}
Any discrete and extreme $\ms A\in\O(\Sigma,\hi)$ and any $\ms B\in\O(\Sigma',\hi)$ are coexistent if and only if they are jointly measurable.
\end{corollary}

\begin{proof}
We use the assumptions and notations of Theorem \ref{lause_ykkonen}. If also ${\rm ran}\,\ms B\subseteq{\rm ran}\,\Mo$, i.e.\ for each $Y\in\Sigma'$ there exists a
$W_Y\in\ov\Sigma$ such that $\Mo(W_Y)=\ms B(Y)$, then (from the proof of Theorem \ref{lause_ykkonen})
$$
\ms B(Y)=\sum_{i=1}^\infty J^*C_i(W_Y)J.
$$
As in the proof of Theorem \ref{lause_ykkonen}, the extremality of $\ms A$ guarantees that $\Sigma'\ni Y\mapsto C_i(W_Y)\in\mc L(\ms P_i\mc K)$ is a POVM. Namely, let $(Y_n)_{n=1}^\infty\subset\Sigma'$ be a disjoint sequence and denote $Y^N:=\bigcup_{n=1}^NY_n$ for any $N=1,\,2,\ldots$. One has
$$
\sum_iJ^*C_i(W_{Y^N})J=\ms B(Y^N)=\sum_{n=1}^N\ms B(Y_n)=\sum_i\sum_{n=1}^NJ^*C_i(W_{Y_n})J
$$
yielding $\sum_iC_i(W_{Y^N})=\sum_i\sum_{n=1}^NC_i(W_{Y_n})$ for all $N$ using the extremality of $\ms A$. Since the images of $C_i$ are supported on mutually orthogonal subspaces, it follows that $C_i(W_{Y^N})=\sum_{n=1}^NC_i(W_{Y_n})$ for all $i$ and $N$, and, again, the sequence of finite partial sums is increasing and bounded from above, and we may continue as in the proof of Theorem \ref{lause_ykkonen} to show that $Y\mapsto C_i(W_Y)$ is $\sigma$-additive. Other properties follow easily. Thus we can define a joint observable of $\ms A$ and $\ms B$ by extending the map
$$
(X_i,Y)\mapsto J^*C_i(W_Y)J
$$
to $\Sigma\otimes\Sigma'$.
\end{proof}

\begin{example}
Consider the case $\hi=\C^3$. Let 
$$
\ms A=(|1\rangle\langle1|,|2\rangle\langle2|,|3\rangle\langle3|)\qquad \ms B=(|\vp^+\rangle\langle\vp^+|,|\vp^-\rangle\langle\vp^-|,|3\rangle\langle3|),
$$
where $\vp^\pm=\frac{1}{\sqrt2}(|1\rangle\pm|2\rangle)$. These observables do not commute and, since they are projection valued, they are not jointly measurable. Despite this fact a 2-valued relabeling of $\ms A$ defined by $\ms A_{\rm rel}=(|1\rangle\langle1|+|2\rangle\langle2|,|3\rangle\langle3|)$ is jointly measurable with $\ms B$ (by Theorem \ref{lause_ykkonen}). This shows that in Corollary \ref{5.2.2014 13:05} it is not enough to have an observable which can be relabeled to an extreme discrete observable which is coexistent with $\ms B$ (compare to \cite{Pell2}).
\end{example}

\section{An extreme mother observable}\label{sec_mot}

Let $(\Om,\Sigma)$ and $(\ov\Om,\ov\Sigma)$ be measurable spaces, $\ms A\in\O(\Sigma,\hil)$, and let $\ms M\in\O(\ov\Sigma,\hil)$ with a minimal dilation $(\ki,\Po,J)$. Assume that $\Mo$ is extreme and $\mr{ran}\,\ms A\subseteq\mr{ran}\,\ms M$. Hence, for any $X\in\Sigma$ there exists a (unique modulo null sets) $Z_X\in\ov\Sigma$ such that $\ms A(X)=\ms M(Z_X)=J^*\ms P(Z_X)J$. Let us denote the (unique) map $X\mapsto\ms P(Z_X)$ by $\ms Q$. It is easy to show that $\ms Q:\,\Sigma\to\li(\ki)$ is a sharp POVM. For example, for any disjoint sequence $(X_i)_{i=1}^\infty$ one has
$$
J^*\left[\ms Q\big(\cup_i X_i\big)-\sum_i\ms Q(X_i)\right]J=\Ao\big(\cup_i X_i\big)-\sum_i\Ao(X_i)=0
$$
implying $\ms Q\big(\cup_i X_i\big)-\sum_i\ms Q(X_i)=0$ since $\ms M$ is extreme; the fact that $\sum_i\ms Q(X_i)$ exists as a weak limit of finite partial sums $\sum_{j=1}^N\ms Q(X_j)=\ms Q\big(\cup_{j=1}^NX_j\big)$ is proven in the same way as in the proof of Theorem \ref{lause_ykkonen}. Now ${\rm ran}\,\ms Q\subseteq{\rm ran}\,\ms P$ and the following theorem is a direct consequence of \cite[Theorem 3.5]{JePu07}.\footnote{To be exact, in order to directly apply \cite{JePu07} in this context, we would have to assume that $\mc K$ is separable. However, a careful reader finds out that it suffices that there be a $\sigma$-finite measure $\mu:\ov\Sigma\to[0,\infty]$ such that $\mu(Z)=0$ if and only if $\ms P(Z)=0$ and the measure $\mu=p_\rho^\Mo$, where $\rho\in\mc S(\hil)$ is a faithful state, has this property. Hence, we do not have to assume that the dilation space $\mc K$ is separable.}

\begin{theorem}\label{th2}
Let $(\Om,\Sigma)$ be standard Borel, $(\ov\Om,\ov\Sigma)$ a measurable space, $\ms A\in\O(\Sigma,\hil)$, and let $\ms M\in\O(\ov\Sigma,\hil)$ be extreme. Then
$\mr{ran}\,\ms A\subseteq\mr{ran}\,\ms M$ if and only if $\ms A$ is a relabeling of $\ms M$.
\end{theorem}

In the next corollary we show that if coexistent observables have an extreme mother observable then they are jointly measurable.

\begin{corollary}
Let $(\Om,\Sigma)$ and $(\Om',\Sigma')$ be standard Borel spaces, $(\ov\Om,\ov\Sigma)$ a measurable space, $\ms A\in\O(\Sigma,\hil)$ and $\ms B\in\O(\Sigma',\hil)$. Let
$\ms M\in\O(\ov\Sigma,\hil)$ be extreme.  If $\mr{ran}\,\ms A\cup\mr{ran}\,\ms B\subseteq\mr{ran}\,\ms M$  then $\ms A$ and $\ms B$ are jointly measurable.
\end{corollary}

\begin{proof}
It follows from Theorem \ref{th2} that there exists measurable functions $f:\,\ov\Om\to\Om$ and $g:\,\ov\Om\to\Om'$ so that $\Ao(X)=\Mo\big(f^{-1}(X)\big)$ and
$\ms B(Y)=\Mo\big(g^{-1}(Y)\big)$ for all $X\in\Sigma$ and $Y\in\Sigma'$, so that, especially, $\ms A$ and $\ms B$ are smearings of $\ms M$. According to the second paragraph of Definition \ref{def:smearing}, this means that $\ms A$ and $\ms B$ are jointly measurable.
\end{proof}

Let $\ms A\in\O(\Sigma,\hil)$ and $\ms M\in\O(\ov\Sigma,\hil)$. Assume that $\nu:\Sigma\to[0,\infty]$ and $\mu:\ov\Sigma\to[0,\infty]$ are $\sigma$-finite measures such that $\mu$ and $\ms M$ (resp.\ $\nu$ and $\ms A$) are mutually absolutely continuous. Indeed, for simplicity, let us choose the measures $\mu=p_\rho^\Mo$ and $\nu=p_\rho^\Ao$ where $\rho\in\mc S(\hil)$ is a faithful state. We keep this state fixed for the rest of this section.

Suppose that $\ms A$ be a post-processing of $\ms M$ given by a  Markov kernel $\beta:\ov\Om\times\Sigma\to[0,1]$ (with respect to $\mu$). Denote by $L^\infty(\nu)$ the von Neumann algebra of (equivalence classes of) $\nu$-essentially bounded $\nu$-measurable complex functions on $\Omega$ and by $L^\infty(\nu)^+$ the positive elements of $L^\infty(\nu)$. Now $\beta$ induces a positive linear map $\beta^*:\,L^\infty(\nu)\to L^\infty(\mu)$ via $[\beta^*(f)](z)=\int_\Om f(x)\beta(z,\d x)$, such that
$$
\ms A(f):=\int_\Om f\,\d\ms A=\ms M\big(\beta^*(f)\big),\qquad f\in L^\infty(\nu)
$$
where $\ms A:\,L^\infty(\nu)\to\mc L(\hil),\;f\mapsto\ms A(f)$ is a positive (normal and unital) linear map. Hence, $\ms A\big(L^\infty(\nu)^+\big)\subseteq\ms M\big(L^\infty(\mu)^+\big)$. The next theorem shows that when the value space of $\ms A$ is standard Borel  the converse holds if $\ms M$ is extreme.

\begin{theorem}\label{prop:extpositive}
Let $(\Om,\Sigma)$ be standard Borel and $(\ov\Om,\ov\Sigma)$ be a measurable space. Let $\ms A\in\O(\Sigma,\hil)$ and $\ms M\in\O(\ov\Sigma,\hil)$ and define $\mu=p_\rho^\Mo$ and $\nu=p_\rho^\Ao$. If $\ms A\big(L^\infty(\nu)^+\big)\subseteq\ms M\big(L^\infty(\mu)^+\big)$ and $\ms M$ is extreme, then $\ms A$ is a post-processing
of $\ms M$.
\end{theorem}

\begin{proof}
Suppose that $(\mc K,\ms P,J)$ is a minimal Na\u{\i}mark dilation for $\ms M$. According to the claim, for any $f\in L^\infty(\nu)^+$ there is a (unique) $g_f\in L^\infty(\mu)^+$ such that $\ms A(f)=\ms M(g_f)$. Let $f_1,\,f_2\in L^\infty(\mu)^+$ and $\lambda\geq0$. The condition $\ms A(f_1+\lambda f_2)=\ms A(f_1)+\lambda\ms A(f_2)$ yields
$$
J^*\big(\ms P(g_{f_1})+\lambda\ms P(g_{f_2})-\ms P(g_{f_1+\lambda f_2})\big)J=0.
$$
The extremality of $\ms M$  implies that $\ms P(g_{f_1+\lambda f_2})=\ms P(g_{f_1})+\lambda\ms P(g_{f_2})$. In the same way, we easily see that the map $f\mapsto g_f$ is
unital, i.e.,\ the constant function $1$ on $\Om$ is mapped to the constant function $1$ on $\ov\Om$. Hence, we may extend the function $f\mapsto\ms P(g_f)$ to a linear map defined on the whole of $L^\infty(\nu)$. Let us denote this map by $\ms Q$. Clearly, $\ms Q$ is positive. Assume that $(f_\alpha)_{\alpha\in A}$ is an increasing net of real-valued elements of $L^\infty(\mu)$ that is bounded from above and denote its supremum by $f$. The net $\big(\ms Q(f_\alpha)\big)_{\alpha\in A}$ is also an increasing net bounded from above by $\ms Q(f)$. The equalities $J^*\sup_\alpha\ms Q(f_\alpha)J=\sup_\alpha\ms A(f_\alpha)=\ms A(f)=J^*\ms Q(f)J$ yield $J^*\big(\sup_\alpha\ms Q(f_\alpha)-\ms Q(f)\big)J=0$ and extremality implies that $\sup_\alpha\ms Q(f_\alpha)=\ms Q(f)$. Hence, $\ms Q$ is normal.

Since $\mu$ and $\ms M$ are mutually absolutely continuous, the spectral measure $\ms P$ is a normal *-isomorphism of $L^\infty(\mu)$ onto its image (with a normal inverse). Hence, the map $L^\infty(\nu)\ni f\mapsto g_f\in L^\infty(\mu)$ is also a normal positive unital linear map. Such a map is a transpose of a norm-continuous predual map $L^1(\mu)\to L^1(\nu)$ that maps positive elements of $L^1(\mu)$ with $L^1$-norm 1 to positive elements of $L^1(\nu)$ with $L^1$-norm 1. In our context, such a map is described by a Markov kernel $\beta:\ov\Om\times\Sigma\to[0,1]$ \cite[Theorem 2.1]{JePu} and, hence, $g_f=\beta^*(f)$ for all $f\in L^\infty(\nu)$.
\end{proof}

If in the claim of Theorem \ref{prop:extpositive} $(\ov\Om,\ov\Sigma)$ is also standard Borel, then $\ms A$ and $\ms M$ are jointly measurable. This fact is contained in the following corollary. In the claim of Theorem \ref{prop:extpositive}, we could have equally well assumed that $\mu:\ov\Sigma\to[0,\infty]$ (resp. $\nu:\Sigma\to[0,\infty]$) is a $\sigma$-finite measure such that $\mu$ and $\ms M$ (resp. $\nu$ and $\ms A$) are mutually absolutely continuous. A similar generalization could be done in the following corollary.

\begin{corollary}\label{cor:extpositivejoint}
Suppose that $(\Om,\Sigma)$ and $(\Om',\Sigma')$ are standard Borel spaces, $(\ov\Om,\ov\Sigma)$ is a measurable space, $\ms A\in\O(\Sigma,\hil)$, $\ms B\in\O(\Sigma',\hil)$, and $\ms M\in\O(\ov\Sigma,\hil)$. Define $\mu=p_\rho^\Mo$, $\nu=p_\rho^{\ms A}$ and $\nu'=p_\rho^{\ms B}$. If $\ms A\big(L^\infty(\nu)^+\big)\cup\ms B\big(L^\infty(\nu')^+\big)\subseteq\ms M\big(L^\infty(\mu)^+\big)$ and $\ms M$ is extreme then $\ms A$ and $\ms B$ are jointly measurable.
\end{corollary}

\begin{proof}
According to Theorem \ref{prop:extpositive}, there is a Markov kernel $\beta:\ov\Om\times\Sigma\to[0,1]$ (resp.\ $\gamma:\ov\Om\times\Sigma'\to[0,1]$) such that $\ms A$
(resp.\ $\ms B$) is a smearing of $\ms M$ by means of $\beta$ (resp.\ $\gamma$). According to the second paragraph of Definition \ref{def:smearing}, $\ms A$ and $\ms B$ are jointly measurable.
\end{proof}

\section{Conclusions}

In the case where the range of an observable $\ms A$ is contained in the range of another observable $\ms M$, there is no guarantee that one can recover the measurement statistics of $\ms A$ from that of $\ms M$. Formally this means that although $\ms A(X)=\ms M(Z_X)$ for some set function $X\mapsto Z_X$ one cannot determine the form of the function. In this work we have shown that if either  $\ms A$ or $\ms M$ is extreme there is a way to connect their statistics. Moreover, we have demonstrated that two coexistent observables $\ms A$ and $\ms B$ are jointly measurable if one of them is extreme and discrete or their mother observable is extreme. We leave it as an open question if the assumption on discreteness can be dropped.

As in \cite{2_coex} we have introduced a concept of joint measurability of binarizations which is a true generalization of coexistence. Namely, we have presented an example of two observables which are not coexistent but their binarizations are jointly measurable. One might ask if similar results hold for coexistence and joint measurability of binarizations in the case of extreme observables as those which we have presented for joint measurability and coexistence. E.g.,\ if we have two observables whose binarizations are jointly measurable,
does it follow that they are coexistent if one of them is extreme? This is left as an open question.
\\

\noindent {\bf Acknowledgments.} The authors thank  Dr.\ Teiko Heinosaari for useful discussions and comments on the manuscript. EH would like to thank Alfred Kordelin Foundation for financial support. RU acknowledges support from Finnish Cultural Foundation. 


\end{document}